\documentclass[twoside,english]{article}
\usepackage{mathpazo}
\usepackage{helvet}
\usepackage{courier}
\usepackage[T1]{fontenc}
\usepackage[latin9]{inputenc}
\usepackage[a4paper]{geometry}
\usepackage{color}
\usepackage{babel}
\usepackage{mathrsfs}
\usepackage{amsmath}
\usepackage{amsthm}
\usepackage{amssymb}
\usepackage{setspace}
\usepackage[authoryear]{natbib}
\usepackage[unicode=true,
 bookmarks=true,bookmarksnumbered=true,bookmarksopen=true,bookmarksopenlevel=1,
 breaklinks=true,pdfborder={0 0 1},backref=false,colorlinks=true]
 {hyperref}
\hypersetup{pdftitle={The LyX Tutorial},
 pdfauthor={LyX Team},
 pdfsubject={LyX-documentation Tutorial},
 pdfkeywords={LyX, documentation},
 linkcolor=black, citecolor=black, urlcolor=blue, filecolor=blue,pdfpagelayout=OneColumn, pdfnewwindow=true, pdfstartview=XYZ, plainpages=false}

\makeatletter
\numberwithin{equation}{section}
\numberwithin{figure}{section}
\theoremstyle{plain}
\newtheorem{thm}{\protect\theoremname}
\theoremstyle{definition}
\newtheorem{defn}[thm]{\protect\definitionname}
\theoremstyle{plain}
\newtheorem{lem}[thm]{\protect\lemmaname}
\theoremstyle{plain}
\newtheorem{cor}[thm]{\protect\corollaryname}
\theoremstyle{plain}
\newtheorem{prop}[thm]{\protect\propositionname}
\theoremstyle{plain}
\newtheorem*{fact*}{\protect\factname}

\makeatother

\providecommand{\corollaryname}{Corollary}
\providecommand{\definitionname}{Definition}
\providecommand{\factname}{Fact}
\providecommand{\lemmaname}{Lemma}
\providecommand{\propositionname}{Proposition}
\providecommand{\theoremname}{Theorem}

\begin{document}
\title{Causal Processes in C{*}-Algebraic Setting}
\author{Chrysovalantis Stergiou}
\maketitle
\begin{center}
Department of History, Philosophy and the Ancient World. School of
Liberal Art and Sciences.
\par\end{center}

\begin{center}
The American College of Greece - Deree
\par\end{center}

\begin{center}
cstergiou@acg.edu
\par\end{center}
\begin{abstract}
In this paper we attempt to explicate Salmon's idea of a causal process,
as defined in terms of the mark method, in the context of C{*}-dynamical
systems. We prove two propositions, one establishing mark manifestation
infinitely many times along a given interval of the process, and,
a second one, which establishes continuous manifestation of mark with
the exception of a countable number of isolated points. Furthermore,
we discuss how these results can be implemented in the context of
Haag-Araki theories of relativistic quantum fields in Minkowski spacetime.
\end{abstract}

\section{Introduction}

Process theories of causation emerged in the first half of the twentieth
century after the advent of the theory of relativity. Reichenbach's
concept of a real sequence, as distinguished from an unreal one in
terms of the method of mark \citep{key-1}, is considered the first
attempt to define \textit{causal processes} as distinguished from
\textit{pseudoprocesses,} while Russell's theory of causal lines \citep{key-2}
is a way to talk about processes in terms of regularities. The real
boost, however, to this approach has been given by the work of W.
Salmon \citep{key-3,key-4} and P. Dowe \citep{key-5,key-6} towards
the end of the century.

Process theories try to explain the cause and effect relation in terms
of mechanisms of generation of causal influence, \textit{causal interactions,}
and of mechanisms of propagation of this influence, \textit{causal
processes. }Both mechanisms are considered spacetime entities, admitting
a geometric representation: causal interactions are localized, they
occur at a place at a time, and they are represented, ideally, by
points of spacetime, while causal processes are locally extended entities
represented geometrically by spacetime curves. Along these curves
the causal influence generated at a spacetime point is propagated.

Various attempts have been made to describe these causal mechanisms
in terms of physical theories and to consider their existence contingent
upon the truth of these theories. Salmon's early theory, which dates
back to 1984, is based on the idea of a mark produced in a process
by a single local intervention, at a spacetime point. The generation
of a mark is considered the prototype of a causal interaction. The
ability of a causal process to transmit a mark at each point of its
spacetime curve, without further assistance, indicates its ability
to transmit causal influence and distinguishes it from a pseudoprocess\textit{.
}Pseudoprocesses, on the other hand, although they exhibit some kind
of uniformity along their geometric representation, they lack the
ability of mark transmission. The motion of a ball in space illustrates
the notion of a causal process while the hit of a bat imparting in
the ball momentum illustrates a causal interaction. On the other hand,
the motion of a shadow or of a light spot qualifies as a pseudoprocess
since any change produced locally at these objects is not propagated,
without further assistance, throughout their history. The connection
between the theory of relativity and this account rests on the impossibility
of a mark to be transmitted at spacelike distance and the consideration
of spacelike curves as possible geometric representations of pseudoprocesses.
Hence, any uniformity manifesting itself in a sequence of events (process)
that \textit{can be }represented geometrically by a spacelike curve
does not qualify as a causal process. In 1997, Salmon, reconsidering
the early version of his theory, has admitted that it provides a \textit{geometric}
account of causation rather than a \textit{physical} one, since no
specification of the nature of causal influence in terms of physical
magnitudes has been given. The recognition of this limitation led
Salmon to the abandonment of the mark method and the specification
of the nature of causal influence, first, in terms of relativistic
invariant quantities, such as the rest mass and the electric charge,
and, later, in terms of conserved quantities (energy, 3-momentum,
etc.), under the influence of Dowe's Conseved Quantity Theory. According
to the later versions of Salmon's theory, a causal interaction is
a local exchange of an invariant or conserved quantity, while a causal
process has the ability to transmit continuously these quantities
from one spacetime point to another in the absence of any further
assistance.

What remained unaltered in the development of Salmon's thought was
his commitment to the idea of spacetime continuity in the transmission
of causal influence by a causal process: once produced, causal influence
is made manifest, either in the form of a mark or in that of a constant
amount of an invariant and/or conserved quantity, at each spacetime
point of the continuous curve. This idea, Salmon claims, provides
a solution to Hume's celebrated problem of causation. It yields the
connection between cause and effect that Hume was unable to trace
at the empirical level. Hume attempted to explain the connection by
interpolating further intermediate causes and effects between the
distant events, yielding, thus, the image of a causal chain connecting
the distant cause with its effect. Nevertheless, as Salmon has pointed
out, this view just multiplies the instances of the problem instead
of solving it, since, now, the connection between each pair of successive
links in the causal chain needs to be accounted for. On the other
hand, were to abandon the discrete picture of distant links in a chain
in favor of a continuous connection, the problem would be solved,
since in a continuum, points are densely distributed and it is meaningless
to talk about the successor of a point. In this way, Salmon developed
the idea of a continuous causal connection between cause and effect.
On the basis of his causal notions, the causal relata are distant
causal interactions connected by a causal process which manifests
continuously, i.e. at each spacetime point of the curve that represents
it geometrically, either the same mark, when the process is marked
at the original interaction, or a constant amount of an invariant
and/or conserved quantity produced at the original interaction, in
the absence of any further interaction.

Dowe's Conserved Quantity Theory, of 1992, is somewhat different from
Salmon's process theory of causation. Firstly, while a causal process
for Salmon, transmits a fixed, non-zero amount of a conserved quantity,
in the absence of interactions, Dowe allows variations of the amount
of conserved quantity manifested along a causal process which needs
not to be non-zero. These variations are either due to causal interactions
or they might be explained in terms of distant action. Thus, although
Salmon's account is firmly committed to action by contact, Dowe's
view leaves room for action at a distance as well. Secondly, Dowe
defines a causal process in terms of objects having identity through
time, while Salmon does not presuppose any account of identity of
an object since it defines causal processes in terms of continuous
manifestations of properties in the absence of external assistance.
Thirdly, the continuity requirement is dropped by Dowe, since he claimed
that an object defining a causal process in terms of its possession
of a conserved quantity may not have continuous existence; admitting,
thus, ``gappy or discontinuous'' causal processes \citep[pp.118-119]{key-6}.
This is a very important difference between the two views, given the
great philosophical importance that Salmon attributes to spacetime
continuity. Fourthly, the explication of the concepts of cause and
effect in terms of causal processes and interactions is different
for Salmon and Dowe. But we need not to further any more about the
differences between these two accounts since in this paper we focus
on the early account of Salmon's theory.

In the following, we attempt to define causal processes in terms of
the mark method, in C{*}-algebraic framework. Firstly, we define a
process to be a C{*}-dynamical system and we identify an instance
of a process in terms of a normal state and an observable (self-adjoint
operator) of the algebra. Secondly, we consider a state transformation,
defined by a non-selective operation on the algebra, as an explication
of a marking interaction of a process. Although all definitions in
this paper do not distinguish between classical and quantum processes,
the results we obtained are restricted by the use of a particular
non-selective operation, suitable only for quantum processes, the
non-selective Lüders measurement. Thus, we use a well-known lemma,
\ref{lem:8}, to provide a necessary and sufficient condition for
the manifestation of the mark in terms of commutation relations between
the projection employed for the marking and the observable (projection)
in which the mark appears (see Cor.\ref{Cor:9}). Fourthly, a causal
process is defined in terms of the ability of a process to transmit
a mark continuously over an interval of values of the parameter of
the one-parameter group of the dynamical system. A second definition,
explicates a weaker notion of continuity of mark transmission, which
demands continuous manifestation of the mark with the exception of
a countable number of isolated points. In this way, we intend to capture
the idea of a 'gappy' process with respect to mark transmission. Then,
we prove two propositions: the first states that for a marking operation
defined in terms of a non-selective Lüders measurement of a given
projection, a mark produced at an early stage of a process, and appearing
at later stage, will appear also in infinite number of intermediate
stages,\ref{prop:11}. The second, states that if we, further, impose
a condition of analyticity for the dynamics of the system on the projection
measured, then 'gappy' processes, in the sense explained, exist, Prop.\ref{prop:13}.

At this point we need to clarify our position regarding these 'gappy',
with respect to mark transmission, processes i.e., those processes
which, if marked, exhibit discontinuities in the mark manifestation,
at particular isolated points along the process. The picture we form
for this sort of a process is that, if marked, it consists of denumerable
parts of continuous transmission of the mark which are delimited by
points of mark disappearance (exempting the upper end of the interval
within which mark transmission occurs). What, then, one would think
of this sort of processes? Would one still be willing to consider
such processes, causal? And, what about the miraculous disappearance
and reappearance of causal influence along a process? Certainly, Salmon
would resist heavily to disregard such brute facts of non-manifestation
of mark and, still, consider the process, causal. It would be as if
he accepted a restoration of the image of a causal chain connecting
the distant cause with its effect, with the links of the chain being
extended over an interval of the process, manifesting continuous transmission
of mark. Thus, were to accept such processes as causal, he would retreat
from a very important contribution of his causal theory, namely, to
deal with the problem of the Humean 'connexion' between the cause
and the effect in terms of the continuity in the transmission of causal
influence. On the other hand, as we have already explained before,
it seems that Dowe does not have a great problem to accept the causal
nature of such processes, despite his denial of the mark method account
of causal processes. 

Far from adhering to Dowe's view, our opinion is that causal processes
should transmit \textit{continuously} causal influence. And the desideratum
of this paper was to explore the possibility of a process that has
the ability to transmit a mark, hence, causal influence, continuously\textit{.
}However, given the restrictions induced by the definitions we have
put forward, the formal results we have obtained, in the C{*}-algebraic
framework, were the best we could do: we preserved some intuition
of continuity, without being able to exclude completely gaps in mark
transmission along the process. So, we suggest the reader to consider
our account a preliminary, possibly insufficient, step in establishing
the possibility of causal processes in the quantum domain. To be sure
to leave open the challenge for the more intuitively plausible continuous
\textit{causal processes}, we dub these 'gappy' processes, \textit{causal-up-to-a-countable-set-of-isolated-points
}(\textit{CSIP-causal})\textit{.}

However, if discussion about causal processes ended here, one would
be entitled to criticize this approach as an attempt to strip process
theories of causation off their mettle, their local character, the
conception of causal processes and interactions as entities that have
a life in space and time. Thus, in the third section of this paper
we attempt to embed the abstract mathematical approach of the second
section, in the context of local quantum physics - in particular,
in that of a Haag-Araki theory of relativistic quantum fields on Minkowski
spacetime - where one can meaningfully be talking about spacetime
entities. In this context, we will be able to consider marking operations
locally defined in terms of local projections and commutation relations
implied by causality-locality axiom of the theory. Prop.\ref{prop:11}
is valid for local projections and it establishes the manifestation
of the mark in infinitely many regions that result from the timelike
translation of a given bounded spacetime region. However, as it is
shown in the Appendix, analytic projections do not belong in local
algebras in a Haag-Araki theory and they cannot be taken to determine
local operations that induce local marking interactions. Thus, in
order to obtain, a similar result about \textit{CSIP-causal} processes,
in terms of Prop.\ref{prop:13}, we should abandon local marking interactions,
and talk about almost local marking operations and approximate manifestations
of the mark. 

\section{Causal Processes and C{*}-Dynamical Systems}

Consider a C{*}-dynamical system\textit{ }$\left\langle \mathcal{A},\mathbb{R},\tau\right\rangle $,
defined in terms of a von Neumann-algebra $\mathcal{A\subseteq\mathcal{B}}(\mathcal{H})$,
where $\mathcal{H}$ is a complex separable Hilbert space, the group
$\left\langle \mathbb{R},+\right\rangle $ of real numbers with the
operation of addition, and a strongly continuous homomorphism of $\mathbb{R}$
into the group of automorhisms of $\mathfrak{\mathcal{A}}$, $\tau:\mathbb{R}\rightarrow Aut(\mathfrak{\mathcal{A}}):t\mapsto\tau_{t}$,
induced by a strongly continuous unitary representation $\mathrm{U}$
of $\left\langle \mathbb{R},+\right\rangle $ on $\mathcal{H}$,

\begin{equation}
\tau_{t}(X)=\mathrm{U}_{t}X\mathrm{U}_{t}^{-1}.\label{eq:1-1}
\end{equation}

\begin{defn}
\label{def:process} A process is a C{*}-dynamical system $\left\langle \mathcal{A},\mathbb{R},\tau\right\rangle $.

The classical or quantum nature of the process depends on whether
the algebra $\mathcal{A}$ is or is not commutative.

For Salmon, processes are particular local entities identified both
in terms of their geometric specifications in spacetime and the empirical
manifestation of their uniformity. Yet, in Def.\ref{def:process}
there is nothing suggestive of a process' particular nature. We suggest
a non-geometric account of a process which makes no reference to spacetime
and locality, so the particular nature cannot be couched in terms
of geometric specifications. In addition, although the uniformity
of a process is reflected on the action of the group of isometries
$t\mapsto\tau_{t}$ on the algebra that determines the dynamical evolution
of the physical observables, a process, according to Def.\ref{def:process},
is not associated with any set of measurements of any observable along
the process, as the specification of a process in terms of the empirical
manifestation of its uniformity would seem to require. To consider
the empirical manifestation of the uniformity along a process, we
need to specify a particular normal state $\omega$ of the algebra
$\mathcal{A}$ and a particular observable $Q\in\mathcal{A}$. Then,
the mapping $t\mapsto\omega(\tau_{t}(Q))$ describes the exhibited
uniformity of the process. Parameter $t$ indexes the stages of the
process - be they spacetime points along a curve, in Salmon's geometric
view - and $\omega(\tau_{t}(Q))$ is the corresponding expectation
values of the observable $\tau_{t}(Q)$, resulting from measurements
performed at each stage of the process. In what follows, we will consider
the observable $Q$ to be a projection and $\omega(\tau_{t}(Q))$
to denote the evolution of the probability of measuring $Q$ along
that process. The restriction to normal states of the algebra allows
for such a probabilistic interpretation of $\omega(\tau_{t}(Q))$
and it is also supported by physical reasons.\footnote{Gleason's theorem and its consequences for a probabilistic interpretation
of a quantum theory and the preparability of a state by means of physically
realizable local operations in an open bounded region of Minkowski
spacetime are some reasons for restricting physically admissible states
to normal states. See, \citep{key-14}}In accordance with above considerations we may define the concept
of an instance of a process: 
\end{defn}

\begin{defn}
\label{def:instance} An \textit{instance} of a process $\left\langle \mathcal{A},\mathbb{R},\tau\right\rangle $
is the quintuple $\left\langle \mathcal{A},\mathbb{R},\tau;\omega,Q\right\rangle $,
where, $\left\langle \mathcal{A},\mathbb{R},\tau\right\rangle $ is
a process, $\omega$ a normal state of the algebra $\mathcal{A}$
and $Q=Q^{*}\in\mathcal{A}$, an observable of the algebra.
\end{defn}

\begin{defn}
\label{def:mark} A\textit{ mark }imposed on\textit{ }the instance
$\left\langle \mathcal{A},\mathbb{R},\tau;\omega,Q\right\rangle $
of a process $\left\langle \mathcal{A},\mathbb{R},\tau\right\rangle $
is a change of the normal state $\omega$ of $\mathcal{A}$ determined
by a non-selective operation $T$ which is manifested as a change
in the expectation value of the observable $\tau_{t_{1}}(Q)$, for
some $t_{1}\in\mathbb{R}.$

Hence, a mark is determined by the state transformation
\end{defn}

\[
\omega\mapsto\omega_{T},
\]

where $T:\mathcal{A}\rightarrow\mathcal{A}$ is a linear map, positive,
completely positive and unital map, i.e.

\[
T(I)=I,
\]

\[
T(X^{*}X)\geq0,
\]
 and the map,
\[
id\otimes T:\mathbb{C^{\mathit{n\times n}}\otimes\mathcal{A\rightarrow\mathbb{C}^{\mathit{n\times n}}\otimes\mathcal{A}}}
\]

is positive for every $n\in\mathbb{N}.$ Thus, the state transformation
is defined as follows:,

\[
\omega_{T}(X)=\omega(T(X))
\]

Moreover, we demand that the state transformation induced by $T$
changes the expectation value of the observable $\tau_{t_{1}}(Q)$,
for some $t_{1}\in\mathbb{R}$, i.e.,

\[
\omega(\tau_{t_{1}}(Q))\neq\omega_{T}(\tau_{t_{1}}(Q)).
\]

Notice, first, that the outcome of marking a process is a new instance
$\left\langle \mathcal{A},\mathbb{R},\tau;\omega_{T},Q\right\rangle $
of the same process which provides a different empirical manifestation
of the uniformity along the process. Secondly, although the notion
of operation has been developed in the context of quantum theory,
Def.\ref{def:mark} does not exclude classical processes. In information
theory, such classical operations are known as classical channels
and they describe manipulation of classical information. They are
defined in terms of linear, positive and unital maps from a commutative
C{*}-algebra to another - since complete positivity is equivalent
to positivity for commutative algebras (see, \citep[p.199]{key-7-1})
- and they induce state transformations that are determined in terms
of classical transition probability distributions in the state space.\footnote{For a finite state space $\mathscr{X}=\left\{ x_{1}...x_{n}\right\} $,
the algebra of observables is $C(\mathscr{X})=\left\{ f:f:\mathscr{X}\rightarrow\mathbb{C}\right\} $
and let $T:C(\mathscr{X})\rightarrow C(\mathscr{X})$ be a classical
operation, such that $g=T(f):\left[\begin{array}{c}
g(x_{1})\\
\vdots\\
g(x_{n})
\end{array}\right]$ =$\left[\begin{array}{ccc}
T_{x_{1}x_{1}} & \ldots & T_{x_{1}x_{n}}\\
\vdots & \ddots & \vdots\\
T_{x_{n}x_{1}} & \cdots & T_{x_{n}x_{1}}
\end{array}\right]$$\left[\begin{array}{c}
f(x_{1})\\
\vdots\\
f(x_{n})
\end{array}\right]$ for every $f\in C(\mathscr{X})$.}

Although the definitions we provide cover both classical and quantum
processes, in what follows, we are going to focus on the exploration
of quantum processes. In particular, the normal state we referred
to in defining the instance of a process, we assume it to be described
in terms of a density operator $W$ on $\mathcal{H},$ i.e. $\omega(\text{\ensuremath{X)=tr(WX)}}$
for every $X\in\mathcal{A}$, while the marking interaction to be
given by the non-selective operation described by the Lüders rule
for a projection $P\in\mathcal{A}$.\footnote{For a brief and very informative presentation of the Lüders rule,
see \citep{key-13}} Hence, a mark is determined by the state transformation

\[
\omega\mapsto\omega_{P},
\]
where $\omega_{P}(\text{\ensuremath{X)=tr(W_{P}X)}}$for every $X\in\mathcal{A}$
and 
\begin{equation}
W_{P}=PWP+(1-P)W(1-P),\label{eq:1}
\end{equation}

and it changes the expectation value of the observable $\tau_{t_{1}}(Q)$,
for some $t_{1}\in\mathbb{R}.$

\[
\omega(\tau_{t_{1}}(Q))=tr(W\tau_{t_{1}}(Q))\neq tr(W_{P}\tau_{t_{1}}(Q))=\omega_{P}(\tau_{t_{1}}(Q)).
\]

As it will become clear soon (see lemma \ref{lem:8}), this way of
restricting the definition of the marking interaction delimits us
to consider only quantum processes, since in classical processes represented
by commutative von Neumann algebras the mark will never appear at
any observable.

Thirdly, our choice of a non-selective operation for the description
of the marking interaction can be justified by reference to to Clifton
and Halvorson's distinction between physical and conceptual operations
on an ensemble of physical systems \citep{key-7-2}. They claimed,
roughly, that given an ensemble of physical systems and a device that
operates on the systems of the ensemble, if the ensemble in the final
state includes all systems from the original ensemble without ignoring
anyone, then the transformation of the ensemble is non-selective.
If, on the other hand, in the final ensemble we keep only those systems
that responded in a certain way to the operation of the device, selecting,
thus, a desired outcome, then we do not consider only the physical
interaction that takes place, but we perform, additionally, a conceptual
operation to result in the final ensemble. Hence, our favoring non-selecting
operations for marking interactions reflects the choice of considering
marking in terms of a physical interaction.

Fourthly, to return to Salmon's views, to mark a process requires
a local interaction that takes place at a spacetime locale. This stipulation
is not made here unless $\mathcal{A}$ is considered the quasilocal
algebra of a local quantum system defined in terms of a net of von
Neumann-algebras indexed by spacetime regions of a background spacetime.
In this case, the state transformation described by the Lüders rule,
(\ref{eq:1}), in terms of a local projection, corresponds to a non-selective
measurement which takes place in a specified spacetime region; hence,
marking a process becomes the outcome of a local interaction. Nevertheless,
as mentioned earlier, these definitions do not intend to include geometric
considerations and to regard causal processes as spacetime entities;
locality issues will be postponed for the next section.

Next, we attempt to explicate the concept of mark transmission by
providing two alternative formulations which make different demands
regarding the continuity of the manifestation of the mark along a
process.
\begin{defn}
\label{def:Trans} For an instance $\left\langle \mathcal{A},\mathbb{R},\tau;\omega,Q\right\rangle $
of a process, a mark induced by a projection $P$ is \textit{continuously
transmitted }over an interval $\left[0,t_{1}\right]$, if and only
if it is manifested as a change in the expectation value of every
observable $Q_{t}=\tau_{t}(Q)$, for $t\in(0,t_{1}]$.

Hence,
\end{defn}

\[
\omega(Q_{t})=tr(WQ_{t})\neq tr(W_{P}Q_{t})=\omega_{P}(Q_{t}),
\]

for every $t\in(0,t_{1}]$.
\begin{defn}
\label{def:CSIP}For an instance $\left\langle \mathcal{A},\mathbb{R},\tau;\omega,Q\right\rangle $
of a process, a mark induced by a projection $P$ is \textit{continuously
transmitted} \textit{up to a countable set of isolated points }over
an interval $\left[0,t_{1}\right]$ if and only if it is manifested
as a change in the expectation value of every observable $Q_{t}=\tau_{t}(Q)$,
for $t\in(0,t_{1}]$, with the exception of a countable set of values
of the parameter $t$ which form a discrete metric space with the
usual metric of the real numbers.
\end{defn}

Hence,

\[
\omega(Q_{t})=tr(WQ_{t})\neq tr(W_{P}Q_{t})=\omega_{P}(Q_{t}),
\]

for every $t\in(0,t_{1}]$ with the exception of a countable set of
values of the parameter which form a discrete metric space with the
usual metric of the real numbers.

According to Def.\ref{def:CSIP}, continuous mark transmission up
to a countable set of isolated points allows denumerable gaps, either
finite or countably infinite, in the manifestation of a mark over
an interval of a process. It is demanded, however, the real values
of the parameter in which the mark disappears to form a discrete space
in the usual topology of the real numbers, i.e. to be isolated. Thus,
for an instance $\left\langle \mathcal{A},\mathbb{R},\tau;\omega,Q\right\rangle $
of a process and a mark induced by a projection $P$, if 
\[
S=\left\{ t:t\in\left[0,t_{1}\right]\quad and\quad\omega(Q_{t})=tr(WQ_{t})=tr(W_{P}Q_{t})=\omega_{P}(Q_{t}),\right\} ,
\]

is the subset of $\left[0,t_{1}\right]$ in which the mark disappears,
Def.\ref{def:CSIP} suggests that for every $t_{0}\in S$, there is
a neighborhood $\mathcal{N}(t_{0},r)=\left\{ t:\left|t-t_{0}\right|<r\right\} $
such that $\mathcal{N}(t_{0},r)\cap S=\left\{ t_{0}\right\} $. Put
more simply, if a mark is continuously transmitted up to a countable
set of isolated points\textit{ }over an interval $\left[0,t_{1}\right]$,
the time instants in which the mark disappears are not densely distributed
over $\left[0,t_{1}\right]$, i.e. they cannot be represented by the
rational numbers in $\left[0,t_{1}\right]$.

Nevertheless, for Salmon, mark transmission over a continuous segment
of a spacetime curve representing geometrically a causal process is
continuous, i.e. the mark appears at each spacetime point of that
segment. This conception is expressed in the following definition,
given that $t$ is taken to be a curve parameter that indexes spacetime
points along that segment:
\begin{defn}
\label{def: causal-process}A process $\left\langle \mathcal{A},\mathbb{R},\tau\right\rangle $
is \textit{causal} if and only if for an instance $\left\langle \mathcal{A},\mathbb{R},\tau;\omega,Q\right\rangle $
of the process there is a mark, induced by a projection $P$, which
can be transmitted continuously over some closed interval of the real
numbers.
\end{defn}

A weaker definition of a causal process would allow 'gappy' causal
processes without eliminating completely the continuity.
\begin{defn}
\label{def:CSIP-process} A process $\left\langle \mathcal{A},\mathbb{R},\tau\right\rangle $
is \textit{CSIP-causal }or \textit{causal-up-to-a-countable-set-of-isolated-points}
if and only if for an instance $\left\langle \mathcal{A},\mathbb{R},\tau;\omega,Q\right\rangle $
of the process there is a mark, induced by a projection $P$, which
can be continuously transmitted up to a countable set of isolated
points over some closed interval of the real numbers.
\end{defn}

At this point, let us say, once more, that although \textit{CSIP-causal}
processes save some part of the original intuition of continuity of
mark transmission, yet, they fail to provide an adequate account of
Salmon's original intuition of continuous propagation of causal influence,
and, without further ado, present some facts regarding the manifestation
of a mark in a process. 
\begin{lem}
\citep{key-7-3} \label{lem:8}For a von Neumann algebra $\mathcal{A}$
and a state transformation $\omega\mapsto\omega_{P},$ induced by
a non-selective Lüders measurement for a projection $P\in\mathcal{A}$,
the expectation value of an observable \textup{$Q\in\mathcal{A}$
}is invariant under this state transformation if and only if \textup{$\,$$\left[P,Q\right]=0.$}
\end{lem}

\begin{cor}
\label{Cor:9} For an instance $\left\langle \mathcal{A},\mathbb{R},\tau;\omega,Q\right\rangle $
of a process $\left\langle \mathcal{A},\mathbb{R},\tau\right\rangle $,
a necessary and sufficient condition for a mark induced by a projection
\textup{$P$} to be manifested at stage \textup{$t_{1}\in\mathbb{R}$}
of the process, at \textup{$\tau_{t_{1}}(Q)$,} is the non-commutation
of $P$ and $\tau_{t_{1}}(Q).$ Equivalently
\[
\omega(\tau_{t_{1}}(Q))=tr(W\tau_{t_{1}}(Q))=tr(W_{P}\tau_{t_{1}}(Q))=\omega_{P}(\tau_{t_{1}}(Q))\Longleftrightarrow[P,\tau_{t_{1}}(Q)]=0.
\]
\end{cor}

\begin{lem}
\textup{\citep{key-8} \label{lem: 10} }Assume we have a strongly
continuous one-parameter group $\mathrm{U}_{t}$ of unitary operators
on a Hilbert space whose generator H has a spectrum bounded from below.
Let $E,F$ be two projections such that

\[
\mathrm{U}_{t}E\mathrm{U}_{t}^{-1}F=F\mathrm{U}_{t}E\mathrm{U}_{t}^{-1}
\]
 for $\left|t\right|<\varepsilon$. If we have $E\cdot F=0$ then
$\mathrm{U}_{t}E\mathrm{U}_{t}^{-1}F=0$ for all $t$.
\end{lem}

In lemma \ref{lem: 10} a spectrum condition is imposed on the dynamics
which is commonly taken as expressing the requirement that the energy
is positive. In addition, spectrum condition in Haag-Araki theories
of relativistic quantum fields is interpreted as further requiring
that effects propagate at velocities less or equal to the velocity
of light in the vacuum. Thus, one may suggest that a causal process
is defined in terms of ``a dynamical system + an adequate spectrum
condition''; a suggestion expressing the idea that causality restricts
the dynamics. First of all, the interpretation of spectrum condition
as a causality condition is not straightforward and unobjectionable.
In a recent article, Earman and Valente conclude that ``... while
it is wrong to regard the spectrum condition as a straightforward
causality condition, it does lie at a node of a complex web of interconnected
causality concepts of AQFT.'' \citep{key-8-1} Secondly, apart from
issues of interpretation of the spectrum condition, this suggestion
regarding causal processes is not directly related to Salmon's early
process theory of causation, which is our source of inspiration in
this paper, and unless one has provided a connection, we will not
consider it here. The following proposition is proven by means of
lemma\ref{lem: 10}, hence, it requires the spectrum condition; however,
as we will see, it is not sufficient to establish a causal process
in Salmon's sense. In this proposition, we assume that the projection
defining the marking operation is identical with the observable which
exhibits the uniformity of the process in the specified instantiation.
\begin{prop}
\label{prop:11} For an instance $\left\langle \mathcal{A},\mathbb{R},\tau;\omega,P\right\rangle $
of a process $\left\langle \mathcal{A},\mathbb{R},\tau\right\rangle $,
consider a mark induced by the projection $P\in\mathcal{A}$. If the
mark is manifested at $P_{t_{1}}=\tau_{t_{1}}(P)$, for some $t_{1}\in\mathbb{R},$
then it will be manifested at $P_{t}=\tau_{t}(P)$, for infinitely
many $t\in\left(0,t_{1}\right]$.
\end{prop}

\begin{proof}
Assume\textit{ }for\textit{ reductio} that for every $t\in\left(0,t_{1}\right)$
the mark disappears. Then, by Cor.\ref{Cor:9}, we infer that $\left[P,P_{t}\right]=0$
for all $t\in\left(0,t_{1}\right)$ which can be extended to $t\in\left(-t_{1},t_{1}\right)$.

Next, by applying lemma \ref{lem: 10} for the pair of orthogonal
projections $P,P^{\perp}=I-P$, where $I$ is the unit of $\mathcal{A}$,
we conclude that $\left[P,P_{t}\right]=0$ for all $t\in\mathbb{R}$;
which is not true, by hypothesis.

Hence, there is a $t_{2}\in\left(0,t_{1}\right)$ in which the mark
appears, i.e. $t_{2}\notin S_{W,P}$. By applying the same argument
for the transmission of the mark over $\left(0,t_{2}\right]$, we
may conclude the existence of a $t_{3}\notin S_{W,P}$, and so on
\textit{ad infinitum.}
\end{proof}
From the proof of Prop.\ref{prop:11} it can be easily understood
that the causal nature of the process has not been established, since
to conclude the manifestation of a mark at the intermediate stages
of the process, we applied the same argument countably many times;
hence we established the manifestation of a mark along an interval
of a process countable infinitely many times but not continuously.
Moreover, the mark is not even transmitted continuously up to a countable
set of isolated points, since one cannot exclude on \textit{a priori}
grounds that the mark appears \textit{only} at time instants $t_{n}=\frac{1}{r}+\frac{1}{nr}$
for some $0\neq r\in\mathbb{\mathbb{R}}$ and for every $n\in\mathbb{N}.$
Thus, we need stronger assumptions, as those stated in Prop.\ref{prop:13},
to meet the conditions of Def.\ref{def:CSIP}, and prove at least
the existence of CSIP-causal processes. But first, let us show another
useful lemma.
\begin{lem}
\label{lem: 12}For an instance $\left\langle \mathcal{A},\mathbb{R},\tau;\omega,Q\right\rangle $
of a process $\left\langle \mathcal{A},\mathbb{R},\tau\right\rangle $,
a necessary and sufficient condition for a mark induced by a projection
\textup{$P$ }not to be manifested at stage \textup{$t_{1}\in\mathbb{R}$}
of the process, at \textup{$\tau_{t_{1}}(Q)$,} is the following:

\[
\omega(\tau_{t_{1}}(Q))=tr(W\tau_{t_{1}}(Q))=tr(W_{P}\tau_{t_{1}}(Q))=\omega_{P}(\tau_{t_{1}}(Q))\Longleftrightarrow\omega([P,[P,\tau_{t_{1}}(Q)])=0.
\]
\end{lem}

\begin{proof}
(By simple calculations). Firstly, for $W_{P}$ defined as in \ref{eq:1},

\begin{equation}
W_{P}\tau_{t_{1}}(Q)=W\tau_{t_{1}}(Q)+PWP\tau_{t_{1}}(Q)-PW\tau_{t_{1}}(Q)-WP\tau_{t_{1}}(Q)+PWP\tau_{t_{1}}(Q).\label{eq:2}
\end{equation}

Next, given the property of the trace, $tr(AB)=tr(BA)$, by properly
grouping the factors in \ref{eq:2}, one may show directly,

\[
tr(W_{P}\tau_{t_{1}}(Q))=tr(W\tau_{t_{1}}(Q))-tr\left\{ [P,[P,\tau_{t_{1}}(Q)]\right\} .
\]
Hence, $\omega(\tau_{t_{1}}(Q))=\omega_{P}(\tau_{t_{1}}(Q))\Longleftrightarrow\omega([P,[P,\tau_{t_{1}}(Q)])=0.$
\end{proof}
As in Prop.\ref{prop:11}, we assume that the projection defining
the marking operation is identical with the observable which exhibits
the uniformity of the proccess in the specified instantiation.
\begin{prop}
\label{prop:13}For an instance $\left\langle \mathcal{A},\mathbb{R},\tau;\omega,P\right\rangle $
of a process $\left\langle \mathcal{A},\mathbb{R},\tau\right\rangle $,
consider a mark induced by the projection $P\in\mathcal{A}$ that
is analytic for $\tau_{t}$.\footnote{See Appendix}The mark is continuously
transmitted up to a countable set of isolated points over an interval
$\left[0,t_{1}\right]$, if it is manifested at $P_{t_{1}}=\tau_{t_{1}}(P)$.
Hence, according to Def. \ref{def:CSIP-process}, $\left\langle \mathcal{A},\mathbb{R},\tau\right\rangle $
is a CSIP-causal process.
\end{prop}

\begin{proof}
The set of analytic elements of $\tau_{t}$ is $\tau_{t}$-invariant
and a {*}-subalgebra of $\mathcal{A}$. Now, since $P$ is analytic
for $\tau_{t}$, the element

\[
P_{-t}\left(P_{-t}P-PP_{-t}\right)-\left(P_{-t}P-PP_{-t}\right)P_{-t},
\]

where $P_{t}=\tau_{t}(P)$, with $P_{0}=P$ is also $\tau_{t}$-analytic.

According to the definition of $\tau_{t}$-analytic elements, there
is a strip $I_{\lambda}$ in the complex plane, $I_{\lambda}=\left\{ z\in\mathbb{C}:\left|Imz\right|<\lambda\right\} $,
and a function $f:I_{\lambda}\rightarrow\mathcal{A}$ such that for

\[
z=t\in\mathbb{R},\quad f(t)=\tau_{t}\left(P_{-t}\left(P_{-t}P-PP_{-t}\right)-\left(P_{-t}P-PP_{-t}\right)P_{-t}\right),
\]

or, equivalently,

\[
z=t\in\mathbb{R},\quad f(t)=\left[P,\left[P,P_{t}\right]\right],
\]

and an analytic function $F:I_{\lambda}\rightarrow\mathcal{\mathbb{C}}$,

\[
F(z)=\omega(f(z)),
\]

for every state $\omega$ of $\mathcal{A}$.

By hypothesis, mark $<\mathcal{A},\tau_{t},\omega,P>$ imposed on
$<\mathcal{A},\tau_{t},\omega>$ for $t=0$, appears at $P_{t_{1}}=\tau_{t_{1}}(P)$.
This fact entails (by lemma \ref{lem: 12}) that $F$ is not a constant
function, i.e.,

\[
F(0)=\omega(f(0))=0\neq\omega\left(f(\tau_{t_{1}}(P))\right)=F(t_{1}).
\]

Now, since $I_{\lambda}$ is an open connected set, by a well-known
theorem of complex analysis, \citep{key-9}(Thm. 1.2 p.90), the set
of zeros of $F$ is at most countable and discrete, i.e. for every
$z_{0}\in I_{\lambda}$, $F(z_{0})=0$, there is a disc $D(z_{0},r)$
of some radius $r>0$ such that for no $z\in D(z_{0},r)$, $z\neq z_{0}$,
$F(z)=0$.
\end{proof}

\section{Causal Processes and Local Quantum Physics}

To discuss the locality aspect of a causal process as defined in the
previous section, one needs to be able to refer meaningfully to local
observables and to local states as well as to formulate locality conditions.
A friendly environment for this type of investigation, which combines
the C{*}-algebraic setting with essential reference to spacetime concepts
is local quantum physics \citep{key-10}. In particular, we will consider
a relativistic quantum field theory on Minkowski spacetime whose models
are structures of the following type:

\[
<\mathcal{H},\,O\mapsto\mathcal{R}(O),\,(a,\mathit{\varLambda})\mapsto U({\displaystyle a,\varLambda})>
\]

where $\mathcal{H}$ is a complex Hilbert space, $O\mapsto\mathcal{R}(O)$
is a net associating a von Neumann algebra $\mathcal{R}(O)$ with
every open bounded region $O$ of Minkowski spacetime, and $(a,\mathit{\varLambda})\mapsto U({\displaystyle a,}\varLambda)$
is a strongly continuous unitary representation of the (proper orthochronous)
Poincaré group in $\mathcal{H}$. The vacuum state is supposed to
be the unique Poincaré invariant state and is represented by a normalized
vector $\varOmega\in\mathcal{H}$ The models are supposed to satisfy
the usual Haag-Araki axioms, \citep{key-11}, so that, in particular,
the global von Neumann algebra $\mathcal{R}$ is well defined and
the vacuum vector $\varOmega$ is cyclic for every local algebra and
separating for every local algebra over a region with nonempty causal
complement\textemdash courtesy of the celebrated Reeh-Schlieder theorem.

Without loss of generality, we restrict our attention to the time
evolution of a system moving along a timelike direction $a$. The
observables measured in any open bounded region $O$ of Minkowski
spacetime are represented by self-adjoint elements of $\mathcal{R}(O)$.
They are evolving with time according to the law:

\[
\mathrm{U_{\mathit{t}}}\mathcal{R}(O)\mathrm{U}_{t}^{-1}=\mathcal{R}(O+ta),
\]

where, $O+ta=\left\{ x\in\mathbb{R^{\mathrm{4}}:\mathit{\mathit{x}\mathit{=y}+a,}\mathit{y}\in\mathit{O}}\right\} $,
$t\in\mathbb{R}$. Respectively, any element $A$ of the global von
Neumann algebra $\mathcal{R}$ is evolving with time to the element
$\mathrm{U_{\mathit{t}}}A\mathrm{U}_{t}^{-1}\in\mathcal{R}$. Hence,
the process we are interested in is the C{*}-dynamical system $\left\langle \mathcal{R},\mathbb{R},\tau\right\rangle $,
where the group of automorphims $\tau_{t}$, defined by \ref{eq:1-1},
describes translations along a given timelike direction $a$. Furthermore,
to define an instance of that process, $\left\langle \mathcal{R},\mathbb{R},\tau;\omega,Q\right\rangle $
we need to single out a normal state $\omega$ and a projection $Q$
of the global algebra $\mathcal{R}$.

In a a relativistic quantum field theory we can say more about the
local features of such a process. Firstly, the geometric representation
of such a process is a timelike line defined by $a$. Further, if
we associate an open bounded spacetime region $O$ with the process,
we may consider the process to be represented geometrically by a tube
defined as the translation of $O$ along the vector $a$. Secondly,
for every $t\in\mathbb{R}$, region $O+ta$ represents geometrically
a stage of the process of which the spacelike complement is well -defined:
\[
\left(O+ta\right)'=int\left\{ x\in\mathbb{R^{\mathrm{4}}:\mathit{\left(\mathit{x}\mathit{-y}\right)^{2}<\mathrm{0},}\mathit{y}\in\mathit{O}+\mathit{ta}}\right\} .
\]

Hence, according to the axiom of causality-locality (\citealp{key-11}:14),
\[
\forall t\in\mathbb{R},\forall S\subseteq R^{\mathrm{4}},\mathrm{open\,and\,bounded,}S\subset\left(O+ta\right)'\Longrightarrow[\mathcal{R}(S),\mathcal{R}(O+ta)]=0
\]

where, two algebras commute if and only if every element of the first
commutes with every element of the second algebra and \textit{vice
versa.}

A common interpretation of the causality- locality axiom is that any
two observables belonging in local algebras of spacelike separated
regions can be measured together, i.e., the measurement of the first
observable does not ``influence'' the measurement of the other -
a fact expressed by the commutativity of the algebras.

In terms of the local algebras $\mathcal{R}(O)$, one may understand
the marking of a process as the result of a local intervention which
takes place in $O$ by demanding the projection $P$ in Eq.\ref{eq:1}
to be a local projection, $P\in\mathcal{\mathcal{R}}(O)$. Moreover,
from the axiom of causality-locality and Lem.\ref{lem:8} it can be
inferred that a mark imposed on an instance $\left\langle \mathcal{R},\mathbb{R},\tau;\omega,Q\right\rangle $
of a process $\left\langle \mathcal{R},\mathbb{R},\tau\right\rangle $
by a local intervention in $O$ is \textit{not} manifested at any
local observable $Q'\in\mathcal{R}(S)$ of any open bounded region
$S\subset O'$. Converely, for any given time interval $I$ of a process,
any change produced in the spacelike complement of the tube representing
the process for the given interval, is not manifested a change in
the characteristic observable $\tau_{t}(Q)$ in the given interval,
$t\in I$. The non-manifestation of a mark at spacelike distance from
the region in which it is produced as well as the shielding of any
part of a process from infuences produced in spacelike distance from
it, builts in this account the demand of the theory of relativity
that no causal influence in propagated at velocities greater than
the velocity of light in the vacuum.

To obtain a causal process one needs to establish the continuity in
the manifestation of the mark along the process. By application of
Prop.\ref{prop:11} we can infer that a mark produced by a local non-selective
Lüders measurement of a projection $P\in\mathcal{\mathcal{R}}(O)$,
in a spacetime region $O$, in an instance $\left\langle \mathcal{R},\mathbb{R},\tau;\omega,P\right\rangle $
of a process $\left\langle \mathcal{R},\mathbb{R},\tau\right\rangle $,
is manifested at infinite time instants in the interval $(0,t_{1})$,
in spacetime regions $O+ta$, at observables $P_{t}=\tau_{t}(P)\in\mathcal{R}(O+ta)$
, given that it is manifested at $P_{t_{1}}=\tau_{t_{1}}(P)\in\mathcal{R}(O+t_{1}a)$
for some $t_{1}\in\mathbb{R}$. Nevertheless, as we have already explained
above, this result is insufficient to justify any intuition of continuity,
thus, to establish the causal character of a process. 

To make, however, stronger demands, as those required by Prop.\ref{prop:13},
we need to abandon some essential locality conditions, since there
are no analytic elements for time translations which belong in a local
algebra associated with any open bounded spacetime region in a Haag-Araki
theory of relativistic quantum fields that satisfy the usual axioms
(see, sec:A.-Appendix). Hence, one is not allowed to assume that
a mark imposed on a process is the result of a local intervention
associated with a state transformation of the form of Eq.\ref{eq:1},
defined in terms of an \textit{analytic} local projection $P.$ Morerover,
given that one cannot localize the marking intervention in any bounded
spacetime region, they cannot talk meaningfully about the manifestation
or not of the mark in spacelike distance from its locus of production,
since the locality-causality condition is valid for local observables.

On the other hand, since analytic elements for time translations form
a norm-dense subset of $\mathcal{R}$, one may take the production
of a mark in an instance $\left\langle \mathcal{R},\mathbb{R},\tau;\omega,P\right\rangle $
of a process $\left\langle \mathcal{R},\mathbb{R},\tau\right\rangle $
to be described in terms of an \textit{analytic} projection $P$ which
differs, in norm, from a local projection $P'\in\mathcal{\mathcal{R}}(O)$
less than a quantity $\delta>0$, 
\[
\left\Vert P-P'\right\Vert <\delta\Longrightarrow\left|\omega(P)-\omega(P')\right|<\delta,
\]

for every state $\omega$; where $\delta$ denotes the experimental
error. Hence, although the \textit{actual }or the \textit{real }physical
operation of marking a process takes place in a spacetime region and
calls for a representation in terms of a local projection, one may
\textit{decide} to represent it \textit{approximately} in terms of
an analytic for time translations element which is practically indistinguishable
from the local projection for every physically admissible (normal)
state of the system. In this way, one may talk of an almost local
mark imposed on $\left\langle \mathcal{R},\mathbb{R},\tau\right\rangle $
in some open bounded region $O$.

Next, consider the manifestation of the mark along the instance $\left\langle \mathcal{R},\mathbb{R},\tau;\omega,P\right\rangle $
of a process $\left\langle \mathcal{R},\mathbb{R},\tau\right\rangle $.
Since $\tau_{t},t\in\mathbb{R}$ is a group of isometries,
\[
\left\Vert P-P'\right\Vert <\delta\Rightarrow\left\Vert \tau_{t}(P)-\tau_{t}(P')\right\Vert <\delta,
\]

and the mark manifestation is required to be detected in $\tau_{t}(P)$,
the time translation of the analytic element $P$, instead of the
time translation of the local projection $P'$ that would represent
the mark, had it been local. Hence, we talk about an almost local
mark associated with a region $O$ which may be approximately manifested
in the translation of $O$, $O+ta$, along the timelike vector $a$.

At this point one may object that, on the one hand, we employed analytic
elements in the representation of an almost local marking operation
claiming that they are practically indistinguishable from local projections
belonging in open bounded spacetime regions - i.e. the difference
in their expectation values, in any state, is undetectable given the
limits experimental error - while on the other, we seem to be satisfied
with Def.\ref{def:mark},\ref{def:Trans} which defines mark manifestation
in terms of a difference in the expectation value of an observable,
without taking into account whether this difference is detectable
or not. To be consistent with the introduction of experimental error
as a limit to the detectability of the difference in the expectation
values of observables in a given state, we should demand

\begin{equation}
\left|\omega(\tau_{t}(P))-\omega_{P}(\tau_{t}(P))\right|=\left|tr(W\tau_{t}(P))-tr(W_{P}\tau_{t}(P))\right|\geq\delta,\label{eq:3.1}
\end{equation}

where $W\in\mathcal{R}$ is a density operator and $W_{P}\in\mathcal{R}$
is defined as in \ref{eq:1}, to claim that the mark is manifested
at $\tau_{t}(P)$ for some $t\in\mathbb{R}$ . The objection is correct
and it requires further investigation to obtain results that would
comply with condition \ref{eq:3.1}. In this paper, however, we follow
Def.\ref{def:mark},\ref{def:Trans} for characterizing a causal process.

Now, by applying Prop.\ref{prop:13} one concludes that a process
$\left\langle \mathcal{R},\mathbb{R},\tau\right\rangle $ represented
geometrically by a tube defined as the translation of $O$ along a
timelike vector $a$ qualifies as a \textit{CSIP-causal process} according
to Def.\ref{def:CSIP-process} since it is capable to transmit continuously
up to a countable set of isolated points a mark defined for an analytic
for time translations projection $P$, imposed on the process almost
locally, in the sense explained, in a bounded spacetime $O$, over
an interval $\left[0,t_{1}\right]$, according to Def. \ref{def:CSIP}.

Valente \citep{key-11-1} has raised serious objections against the
use of analytic elements for time translations to describe operations
in local quantum systems.\footnote{The discussion in this Valente's paper refers to almost local operators
which are constructed as analytic elements for the generators of the
group of translations in Minkowski spacetime.} One objection says that since analytic elements cannot be local in
any standard Haag-Araki theory, they represent unistantiated events
in spacetime and they cannot be taken, in any meaningful way, to be
in relation to any event or to have any local effect in spacetime.The
second objection relates to the way analytic elements can be constructed
by smearing local observables of a bounded region $O$ (see, sec:A.-Appendix),
providing, thus, some type of approximate localization in $O$ (by
disregarding the 'tails' of the smearing function). However, essentially,
analytic elements are global observables and to talk about almost
\textit{local} operations is rather a euphemism. The third objection
says that since analytic observables are not local they are not real
observables and no actual measurement can be performed of these fictional
observables. Hence, no real effect can be produced by such fictional
operations.

Both the first and the third objection fit nicely for the critical
discussion of \citep{key-11-2}, as intended. There, the operation
is defined in terms of an analytic elements and the manifestation
is determined in the change of the expectation value of a local observable;
hence, one may reasonably object that uninstantiated events in spacetime
cannot be in relation with proper events and that fictional operations
cannot produce real effects. In the present discussion, however, for
better or worse, the relation is established between homologous events,
both uninstantiated, which lie at a fictional level. Here, we tentatively
explore what would happen if local operations and local observables
were represented by non-local ones which, nevertheless, bear a special
relation to their local counterparts that allows us to talk about
almost local marking operations and their corresponding manifestations. 

To defend somewhat further the tentative hypothesis of using analytic
elements for the description of operations, let me refer to an old
paper of Steinmann \citep{key-11-3} in which (selective) operations
representing particle detection experiments are defined in terms of
analytic elements. There, Steinmann proved a result that captures
a basic intuition of particle causal processes. Namely, he has proven
that in an one-particle state of a field theory, ``...the probability
that two counters are triggered by the particle is appreciable only
if their separation is approximately parallel to the momentum of the
particle.'' In addition, if a third counter is placed, then ``...all
three counters can be triggered only if they lie roughly in a straight
line.'' Still, one may object to the description of a particle's
motion in terms of locality considerations. Yet, we, surely, would
like to save the intuitive plausibility of such particle motion as
good candidates even of an 'approximate' process. 

As a final comment, I would like to point out the problematic implementation
of relativistic causality when discussing causal processes in terms
of analytic under time translations elements in local quantum physics.
As explained before, the axiom of causality-locality guarantees that
when considering causal processes in terms of elements of local algebras,
the mark imposed on a process is not manifested at spacelike distance
from the locus of its production. One cannot make the same stipulation
when dealing with analytic for time translation projections that are
associated with some spacetime region $O$ in the aforementioned sense,
since the relevant axiom is not applicable any more. Still we may
claim that the mark defined in terms of analytic projection $P$ does
not become manifest at observables which commute with $P$; however,
we cannot any more claim that every observable associated with an
open bounded region in the spacelike complement of $O$, be that a
local one or a practically indistinguishable (within $\delta$) from
a local observable, commutes with $P$, making, thus, the relevant
mark disappear.

\section{Summary}

Process theories of causation employ three types of locality conditions
in the explanation of the cause and effect relation: (a) the production
of causal influence and its manifestation along a process occurs at
spacetime points; (b) the propagation of causal influence is continuous
in spacetime; (c) for any given spacetime point there is a set of
spacetime points which are not connectable by means of causal processes,
hence, these spacetime points are shielded from causal influence produced
at the given point.

In this paper, we attempted to define causal processes in the absence
of any reference to spacetime. We referred to C{*}-dynamical systems
their states and state transformations described in terms of projections
in a von Neumann algebra realizing marking processes (Def. \ref{def:process},\ref{def:instance},\ref{def:mark},\ref{def:Trans},\ref{def: causal-process},\ref{def:CSIP-process}).
Well-known conditions for the shielding of certain observables characterizing
quantum processes from the influence of marking operations, couched
in terms of the invariance of the expectation value of these observables
under the marking state transformations, were formulated (Lem.\ref{lem:8}).
The originality of our contribution rests, mainly, on the attempt
to establish some satisfactory conception of continuity in the propagation
of the mark along such a process. Proper continuity in mark transmission,
required for a causal process could not be established. Nonetheless,
we have shown that a mark is manifested infinitely many times in an
interval of a process (Prop.\ref{prop:11}), and that, on certain
conditions, the mark is manifested coninuously up to a countable set
of isolated point (Prop.\ref{prop:13}), establishing, thus, the existence
of CSIP-causal processes. Thus, in this abstract algebraic setting
we expressed conditions for process theories which if successfully
associated with spacetime properties and relations they would yield
a local account of causal processes.

To obtain a local account of causal processes, we examined a Haag-Araki
theory of relativistic quantum fields on Minkowski spacetime. In this
context, local observables belong in von Neumann algebras associated
with open bounded spacetime regions and a local marking operation
can be defined in terms of a local projection. In addition, we explained
how one can perceive an almost local marking operation, as a representative
of a local marking operation, and an approximate manifestation of
this mark associated with an analytic projection under time translations.
Analytic elements do not belong in the algebra of any open bounded
region but they can be practically indistinguishable from local elements.
In this way, the first of the three aforementioned locality conditions
for causal processes was treated. The axiom of causality-locality
satisfied by Haag-Araki theories provides a necessary and sufficient
condition for the non-manifestation of a local mark at spacetime distance
from the region in which it is generated. A similar condition for
marking operations in terms of analytic elements could not been established
and I am not sure whether it can be assumed independently. On the
other hand, continuity has been established, in the rather weak form
of CSIP mark transmission by Prop.\ref{prop:11} in terms of almost
local marking operations, while if our attention were to be restricted
to local marking operations, the same goal would not be attained.
Although we could not establish the existence of causal processes
in Salmon's sense, yet, we did try to explore, formally, different
approaches which emphasize on aspects of the philosopher's conception
within some limitations.

\medskip{}

\section*{A. Appendix.\label{sec:A.-Appendix} On analytic Elements for the
one-Parameter Group of time Translations in Local Quantum physics}

Given a $\sigma(X;F)$-continuous group of isometries of $X$, an
element $A\in X$ is said to be \textit{analytic for $\tau$} if there
exist a strip 
\[
I_{\lambda}=\left\{ z\in\mathbb{C}:\mid Imz\mid<\lambda\right\} 
\]
 in $\mathbb{C}$, and a function $f=f_{A}:I_{\lambda}\rightarrow X$
such that

(i) $f(t)=\tau_{t}(A)$ for $t\in\mathbb{R},$

(ii) $z\mapsto\eta(f(z))$ is analytic for all $\eta\in F$ and $z\in I_{\lambda}$.

In these conditions we write $f(z)=\tau_{z}(A),\:z\in I_{\lambda}.$
Moreover, if $\lambda=\infty$ we say that $a$ is entirely analytic.
(\citep{key-12}:Def.2.5.20)

The groups of isometries we consider in this paper are $\sigma(\mathcal{A},\mathcal{A}^{*})$-continuous
i.e. for all $A\in\mathcal{A},\,t\mapsto\eta(\tau_{t}(A))$ is continuous
for every $\eta\in\mathscr{\mathcal{A^{*}}},$which amounts to the
requirement that $t\mapsto\tau_{t}(A)$ be continuous in norm for
every $A\in\mathcal{A}$.

To obtain a family of analytic elements for the group of time translations
one may begin with any element localized in some open bounded region
$O$ of Minkowski spacetime, $A\in\mathcal{R}(O)$, and smear its
time translation over $\mathbb{R}$ with a Gaussian function depending
on a parameter $n\in\mathbb{N}.$ Namely,
\begin{center}
$A_{n}=\sqrt{\frac{n}{\pi}}\int_{-\infty}^{+\infty}e^{-nt^{2}}\tau_{t}(A)dt,\quad n\in\mathbb{N}.$
\par\end{center}

Each $A_{n}$ is an entire analytic element for $\tau_{t}$ and the
family $\left\{ A_{n}\right\} _{n\in\mathbb{N}}$ converges to $A$
in the $\sigma(X;F)$ topology. Moreover, in (\citealp{key-12}:Cor.2.5.23)
it is shown that the set of entire analytic elements form $\sigma(\mathcal{R};\mathcal{R^{*}})$-
continuous group of isometries form a norm-dense subset of $\mathcal{R}.$
\begin{fact*}
\textup{(Fredenhagen, 2019)}\footnote{Private communication} In
a Haag-Araki theory of relativistic quantum fields on Minkowski spacetime,
no non-trivial analytic elements for time translations can be localized
in open bounded regions of Minkowski spacetime.
\end{fact*}
\begin{proof}
Let $A$ be analytic under time translations element which is also
localized in a bounded region $O$. Let $B$ be localized in the spacelike
complement of an $\varepsilon$-neighbourhood of $O$. Then the commutator,
\[
\left[\tau_{t}(A),B\right]
\]
vanishes in an open interval of $t$. For $\left|t\right|<\varepsilon$,
we have
\[
\left\langle B^{*}\varOmega,\tau_{t}(A)\varOmega\right\rangle =\left\langle \tau_{t}(A^{*})\varOmega,B\varOmega\right\rangle ,
\]

where $\varOmega\in\mathcal{H}$ is the vacuum vector.

By the spectrum condition, which demands the generator of the time
translations to be positive, the term on the left hand side defines
a function that is analytic in the upper halfplane while the term
on the right hand side, an analytic function in the lower halfplane.
By assumption, they are also analytic in a strip around the real axis.
Since these analytic functions coincide in an open interval of the
real axis, we obtain an entire bounded analytic function which, therefore,
is constant.

By the Reeh-Schlieder Theorem, these observables $B$ generate from
the vacuum $\varOmega$ a dense set of vectors in the Hilbert space
$\mathcal{H}$. Thus, $A\varOmega$ is invariant under time evolution
and, due to the uniqueness of the vacuum,
\[
A\varOmega=\left\langle \varOmega,A\varOmega\right\rangle \varOmega.
\]

Thereby, we get the result that an analytic local observable is a
multiple of the identity of the algebra.
\end{proof}

\paragraph*{\textsc{Acknowledgments:} The author expresses his gratitude to Prof.
Dr. K. Fredenhagen for his substantial help. }

\end{document}